\documentclass [11pt]{article}
\usepackage{graphicx,amssymb,amsfonts,latexsym,amsmath,amsthm}
\usepackage{epsfig}
\usepackage{vmargin}
\usepackage{color}
\usepackage[english]{babel}

\newcommand{\ddt}[1]{\frac{\mathrm{d}#1}{\mathrm{d}t}}
\newcommand{\be}{\begin{equation}}
\newcommand{\ee}{\end{equation}}
\newcommand{\ba}{\begin{equation} \begin{aligned}}
\newcommand{\ea}{\end{aligned} \end{equation}}

\newtheorem{mydef}{Definition}
\newtheorem{thm}{Theorem}
\newtheorem{lem}{Lemma}

\title{\Large \bf Non-Markovian stochastic
epidemics in\\ extremely heterogeneous populations}
\author{T.~House\footnote{Warwick Mathematics Institute, University of Warwick,
Coventry, CV4 7AL, UK.}}

\begin{document}

\maketitle

\begin{abstract}
A feature often observed in epidemiological networks is significant
heterogeneity in degree. A popular modelling approach to this has been to
consider large populations with highly heterogeneous discrete contact rates.
This paper defines an individual-level non-Markovian stochastic process that
converges on standard ODE models of such populations in the appropriate
asymptotic limit. A generalised Sellke construction is derived for this model,
and this is then used to consider final outcomes in the case where
heterogeneity follows a truncated Zipf distribution.
\end{abstract}

\setcounter{equation}{0}
\section{Introduction}

A wide variety of different mathematical structures are used to aid modern
infectious disease epidemiology, with particular effort expended on capturing
the heterogeneities in contact patterns that influence the spread of
pathogens~\cite{Anderson:1991,Keeling:2007}. One of the most active topics in
mathematical epidemiology is the use of network theory to represent contact
structure~\cite{Bansal:2007,Danon:2011}.

In the influential work by Barab\'{a}si and Albert~\cite{Barabasi:1999}, it was
shown that networks generated with preferential attachment to highly connected
nodes exhibit power-law behaviour: the probability that a node has $k$ links
becomes proportional to $k^{\alpha}$ for large $k$. In the pure
Barab\'{a}si-Albert (BA) model, $\alpha=-3$, although generalisations of this
model allow for $-3<\alpha<-2$~\cite{Durrett:2007}. Such distributions have
divergent second moment, but finite mean.

Work on synthetic populations~\cite{Eubank:2004}, sexual
behaviour~\cite{Schneeberger:2004} and contact surveys~\cite{Danon:2012} does
provide evidence for power-law behaviour; although concerns remain about the
statistical problems associated with looking for power laws in empirical
data~\cite{Clauset:2009}.

A crucial quantity in infectious disease epidemiology is the basic reproductive
ratio, $R_0$, which defines an epidemic threshold such that for $R_0>1$ the
disease is able to infect an appreciable proportion of a large
population~\cite{Diekmann:1990}.  The fact that $R_0$ is proportional to the
second moment of the contact network degree distribution has been known by
mathematical epidemiologists for some time, and is included in the standard
textbooks~\cite{Diekmann:2000}.  Just over a decade ago, work by
Pastor-Satorras and Vespignani noted that the divergence of the second moment
of the degree distribution might lead to divergence of $R_0$ and the absence of
an epidemic threshold~\cite{Pastor-Satorras:2001}, although it was quickly
pointed out by May and Lloyd that this need not apply for finite
systems~\cite{May:2001}.  Nevertheless, it is still clear that epidemics on
finite networks generated using the BA model have very different behaviour from
epidemics in homogeneously mixing
populations~\cite{Pastor-Satorras:2002,Berger:2005}, and interest persists in
the effects of extreme heterogeneity on epidemic
behaviour~\cite{Kiss:2006,Pastor-Satorras:2010}.

This paper defines a non-Markovian stochastic process for an epidemic in a
finite population with a discrete-valued individual-level measure of
heterogeneity in mixing. The Markovian version of this model has as its
large-population limit the ODE models used in the `annealed network'
approximation for heterogeneous networks~\cite{Pastor-Satorras:2010}. A
generalised Sellke construction is derived for this process to enable fast
sampling from the final size distribution. This is used to derive a
visualisation of the asymptotic absence of a threshold seen in highly
heterogeneous models. Note that the proofs presented are intended to be clear
rather than excessively detailed, and so do presuppose some familiarity with
standard results in probability and stochastic processes~\cite{Grimmett:2001}.

\vspace*{0.5cm}
\setcounter{equation}{0}
\section{Methods}

\subsection{A heterogeneous epidemic model}

Consider a population made up of $N$ individuals, indexed by $i,j,\ldots\in
\mathbb{N}$.  Each individual has an integer-valued random variable called the
\textit{degree}, $K_i\sim D$, where $D$ is called the \textit{degree
distribution}, representing its level of mixing with the rest of the
population. We will write $d_k$ for the probability that a random variable with
distribution $D$ takes the value $k$.  To each individual is assigned a
non-negative real-valued infectious period $T_i\sim\Gamma$, where $\Gamma$ is
called the \textit{infectious period distribution}. To each individual is
assigned a random variable $X_i(t)$ representing its disease state, which is
either susceptible $S$, infective $I$, or removed $R$.  Contacts between
individuals $i$ and $j$ happen at the points of a Poisson process with rate
$\tau K_i K_j$.  If a susceptible and infective individual contact each other,
then the susceptible becomes infective. An individual $i$ that becomes
infective at time $t$ is removed at time $t+T_i$~.

\vspace*{0.25cm}
\begin{mydef}
A stochastic process satisfying the description above is called a \emph{Discrete
Heterogeneous Stochastic Epidemic (DHSE)}.
\end{mydef}

Note that to specify a DHSE fully requires two distributions, $D$ and $\Gamma$,
two parameters, $N$ and $\tau$, and a set of of initial states for the
random variables, $\{X_i(0)\}$.

\subsection{Deterministic limit}

\subsubsection{Dynamics}

Consider a continuous-time Markov chain with non-independent integer-valued
random variables representing the numbers of the population susceptible and
infectious of different degree, $\{S_k(t), I_k(t)\}_{k\in \mathcal{K} \subseteq
\mathbb{N}}$.  There are two types of event and rate:
\begin{eqnarray}
(S_k,I_k) & \rightarrow & (S_k-1,I_k+1) \text{ at rate } \tau k S_k
 \sum_{l\in \mathcal{K}} l I_l \text{ ,}
\label{ktr} \\
I_k & \rightarrow & I_k-1 \text{ at rate } \gamma I_k \text{ .}
\label{kre}
\end{eqnarray}
We have $S_k(t) + I_k(t) \leq N_k , \forall t$, and can write $\mathbf{N} =
(N_k)$ for the vector of numbers of individuals of degree $k$. At time $t=0$,
we pick $\mathbf{N} \sim \mathrm{Multi}(\mathbf{d},N)$ where $\mathbf{d}=(d_k)$
and `Multi' is used to stand for the multinomial distribution probability
distribution function. This constrains the permissible choice of initial
conditions $\{S_k(0), I_k(0)\}$.

\vspace*{0.25cm}
\begin{mydef}
A stochastic process satisfying the description above is called a \emph{Discrete
Heterogeneous Markovian Epidemic (DHME)}.
\end{mydef}

Note that to specify a DHME fully requires the distribution $D$, two
parameters, $N$ and $\tau$, and a set of rules for generation of initial states
for the random variables $\{S_k(0), I_k(0)\}$ given $\mathbf{N}$.

\vspace*{0.25cm}
\begin{thm}
A DHSE where $\Gamma = \mathrm{Exp}(\gamma)$, i.e.\ the infectious periods are
exponentially distributed with mean $1/\gamma$, and where $m$ individuals are
selected at random to be initial infectives (with other individuals initially
susceptible) is equivalent to a DHME with the same $D$, $N$, and $\tau$.
\end{thm}
\begin{proof}
Start with a DHSE, and define the lumped variables
\be
S_k(t) := \sum_i \mathbf{1}_{\{X_i(t) = S\ \&\ K_i = k\}} \text{ ,} \qquad
I_k(t) := \sum_i \mathbf{1}_{\{X_i(t) = I\ \&\ K_i = k\}} \text{ .}
\label{XYk}
\ee
The multinomial distribution of $\mathbf{N}$ follows from the definition of the
DHSE and the independence of the $\{K_i\}$. Letting $\mathbf{S} = (S_k)$,
$\mathbf{I} = (I_k)$, if $m$ individuals are selected at random to be the
initial infectives then in the lumped variables this corresponds to a
multivariate hypergeometric distribution $\mathbf{I}(0) \sim
\mathrm{MultiHG}(\mathbf{N},m)$.

In the DHSE we have the rate of transmission given by a Poisson process,
meaning
\be
\mathrm{Pr}(X_i(t+\delta t) = I \; | \; X_i(t) = S) =
- \tau K_i \sum_j K_j \mathbf{1}_{\{X_j(t) = I\}}\delta t
+ o(\delta t) \text{ ,}
\ee
and~\eqref{ktr} is simply a lumping of this equation. In the same way,
exponential distribution of recovery times is equivalent to a death process
on infectives, meaning that
\be
\mathrm{Pr}(X_i(t+\delta t) = R \; | \; X_i(t) = I) = 
\gamma \delta t + o(\delta t) \text{ ,}
\ee
and~\eqref{kre} is a lumping of this process.
\end{proof}

\begin{thm}
\label{odelim}
A DHME with finite $\mathcal{K}$ and $D$ independent of $N$ has its
deterministic large $N$ limit given by an ODE system of the form
\be
\ddt{s_k} = -\beta k s_k \sum_l l \iota_l
\text{ ,} \qquad
\ddt{\iota_k} = \beta k s_k \sum_l l \iota_l - \gamma \iota_k \text{ .}
\label{mfodes}
\ee
\end{thm}
\begin{proof}
Note that the rates~\eqref{ktr} and~\eqref{kre} can be written in the form
\begin{eqnarray}
(S_k,I_k) & \rightarrow & (S_k-1,I_k+1) \text{ at rate } (\tau N) N k
\left[ \frac{S_k}{N} \right]
 \sum_{l\in \mathcal{K}} l \left[ \frac{I_l}{N} \right] \text{ ,} \\
I_k & \rightarrow & I_k-1 \text{ at rate } \gamma N 
\left[ \frac{I_k}{N} \right] \text{ .}
\end{eqnarray}
Defining 
\be
s_k(t) := \mathbb{E}\left[ \frac{S_k}{N} \right] \text{ ,} \qquad
\iota_k(t) := \mathbb{E}\left[ \frac{I_k}{N}  \right] \text{ ,} \qquad
\beta := \tau N \text{ ,} \label{xyk}
\ee
we can then apply the results of Kurtz~\cite{Kurtz:1970,Kurtz:1971} to derive
the deterministic limit~\eqref{mfodes} with corrections appearing at
$O(N^{-1/2})$.
\end{proof}

\subsubsection{Early behaviour}

Analysis of~\eqref{mfodes} using dynamical systems theory shows that
perturbations away from the disease-free equilibrium $(\mathbf{s},
\boldsymbol{\iota}) = (\mathbf{d}, \mathbf{0})$, grow exponentially
with rate
\be
r = \beta {\mathbb{E}_D[K^2]} - \gamma
 \text{ ,} \label{req}
\ee
provided $r$ is positive.  The argument made by Pastor-Satorras and
Vespignani~\cite{Pastor-Satorras:2001} is then essentially that if the second
moment of $D$ is divergent, then~\eqref{req} implies that an arbitrarily small
$\beta$ can still lead to a growing epidemic.

Note that equations like~\eqref{mfodes} have been used for some time in
epidemiological modelling~\cite{May:1988}.  In a more general sense this system
represents a special, discrete version of the proportionate mixing analysed in
the non-Markovian case by~\cite{Diekmann:2000} and so using this analysis the
quantity
\be
\rho := \beta {\mathbb{E}_D[K^2]\mathbb{E}_{\Gamma}[T]}
\label{rhodef}
\ee
will be an appropriate definition for the basic reproductive ratio $R_0$ of the
large-$N$ DHSE: we expect a large epidemic exactly when $\rho>1$. We will make
use of the definition~\eqref{rhodef} in the finite case where this is not a
`threshold' but remains an appropriate quantity to use when comparing different
parameter values.

\subsubsection{Late behaviour}

Following the broad approach of~\cite[Appendix B]{Kiss:2006}, it is possible to
manipulate~\eqref{mfodes} in the limit where the initial level of infection is
very small and obtain the expression
\be
s_k^\infty := \lim_{t\rightarrow\infty}
s_k(t) = d_k \left( \mathrm{exp}\left( - \frac{\beta}{\gamma} \sum_l l \left(
d_l - s_l^\infty \right) \right) \right)^k \text{ .}
\ee
Then the expected final proportion of individuals in the susceptible class at
the end of the epidemic is
\be
s^\infty = G(\psi) \text{ ,}
\qquad \text{for} \quad
\psi = \mathrm{exp}\left( - \frac{\beta}{\gamma} \left(
\mathbb{E}_D[K] - \psi G'(\psi) \right) \right) \text{ ,} \label{psimark}
\ee
where $G$ is the probability generating function of $D$.  Looking at this
transcendental equation, solutions for which $\psi<1$ will only exist if $r>0$
for $r$ as in~\eqref{req}. Further consideration of these equations leads us to
expect a change in behaviour around $\rho=1$, if such a change exists, to be
percolation-like.

\subsection{A generalised Sellke construction}

Consider a population made up of $N$ individuals, indexed by $i,j,\ldots\in
\mathbb{N}$.  Each individual has an integer-valued random variable $K_i\sim
D$.  Each individual also picks a non-negative real-valued infectious period
$T_i\sim\Gamma$, and has a random variable $X_i(t)\in \{S, I, R\}$ representing
its disease state.  Each initially susceptible individual picks a
\textit{resistance} $Q_i \sim \mathrm{Exp}(K_i)$, with initial infectives
having zero resistance. The \textit{infectious pressure} of the epidemic at
time $t$ is defined as
\be
\Lambda(t) := \tau \int_0^t \sum_{j=1}^{N}
K_j \mathbf{1}_{\{X_j(u) = I\}} \;du \text{ .}
\label{ladef}
\ee
An initially susceptible individual becomes infectious at the first time $t$
when $Q_i < \Lambda(t)$.

\vspace*{0.25cm}
\begin{mydef}
A stochastic process satisfying the description above is called a \emph{Discrete
Heterogeneous Sellke Construction (DHSC)}.
\end{mydef}

Note that to specify a DHSC fully requires the same quantities as the DHSE: two
distributions, $D$ and $\Gamma$, and two parameters, $N$ and $\tau$, and a set
of initial states for the random variables, $\{X_i(0)\}$.

\vspace*{0.25cm}
\begin{lem}
\label{sellem}
A DHSE and DHSC with the same $D$, $\Gamma$, $N$, $\tau$ and $\{X_i(0)\}$ are
equivalent.
\end{lem}
\begin{proof}
From the conditions stated, the only difference between the DHSE and DHSC is in
the behaviour of transmission. We can follow the broad exposition
in~\cite{Andersson:2000}. First, note that in the DHSC
\be
Q_i \sim \mathrm{Exp}(K_i) \qquad \Rightarrow \qquad
\mathrm{Pr}(Q_i > q) = {\rm e}^{-K_i q} \text{ .}
\ee
Then over a time period $[t, t+\delta t]$ we know that 
\ba
\mathrm{Pr}(Q_i > \Lambda(t+\delta{}t) | Q_i > \Lambda(t)) & =
\mathrm{exp}(-(\Lambda(t+\delta{}t) - \Lambda(t))K_i)\\
& = 1 - \tau K_i \sum_j K_j \mathbf{1}_{\{X_j(t) = I\}}
+ o(\delta{}t) \text{ .}
\ea
In the DHSE, we know from the behaviour of Poisson processes that 
\be
\mathrm{Pr}(X_i(t+\delta{}t)=S | X_i(t)=S) 
= 1 - \tau K_i \sum_j K_j \mathbf{1}_{\{X_j(t) = I\}}
+ o(\delta{}t) \text{ .}
\ee
Therefore the behaviour of transmission in the two models is equivalent.
\end{proof}
\begin{thm}
\label{fsthm}
If Z is the final number of removed individuals in a DHSE
\be
Z := \lim_{t\rightarrow\infty} \sum_{j=1}^N
\mathbf{1}_{\{X_j(t) = R\}}\text{ ,}
\ee
then by picking degrees, resistances and recovery times as in the DHSC, and
defining the indexing such that $Q_{i}\leq Q_{i+1} , \forall i$ ,
\be
Z = \mathrm{min} \bigg\{ i \;\bigg|\; Q_{i+1} > 
\tau \sum_{j=1}^{i} K_j T_j \bigg\} \label{zfs} \text{ .}
\ee
\end{thm}
\begin{proof}
From Lemma~\ref{sellem}, the DHSE is equivalent to the DHSC. Permutation of the
indices can be performed without loss of generality since there is no
dependence of $D$ or $\Gamma$ on $i$. Consideration of the
integral~\eqref{ladef} gives
\be
\int_0^{\infty} \mathbf{1}_{\{X_i(u) = I\}} \;du =
\begin{cases} T_i & \text{ if } i \text{ is ever infected,}\\
 0 & \text{ if } i \text{ is never infected.}\end{cases}
\ee
Then from the choice of indexing,
\be
Q_{i+1} > \tau \sum_{j=1}^{i} K_j T_j =: \Lambda_\infty
\qquad \Rightarrow \qquad
Q_{\jmath'>i} > \Lambda_\infty 
= \lim_{t\rightarrow \infty} \Lambda(t) \text{ .}
\ee
This means that the first individual to resist infection will determine
the final size as in~\eqref{zfs}.
\end{proof}

\vspace*{0.5cm}
\setcounter{equation}{0}
\section{Results and Discussion}

The major practical benefit to Theorem~\ref{fsthm} is the enhanced numerical
performance of the Sellke construction compared to direct
simulation~\cite{House:2013}.  We use this to simulate the final size of a DHSE
where $D$ is a truncated Zipf distribution, i.e.
\be
d_k = \kappa_\alpha k^{\alpha} \text{ ,}
\qquad \text{for} \quad \kappa = \sum_{k=1}^{k_{\mathrm{max}}} k^{\alpha}
\text{ .}
\ee
In the limit $\alpha\rightarrow-\infty$, this reduces the heterogeneous model
to the standard `mass action' epidemic where each individual has $K_i =1$. But for
smaller $|\alpha|$, the $\lceil1-\alpha\rceil{}$-th moment will grow with
$k_{\mathrm{max}}$ rather than converging.  For $-3< \alpha \leq -2$, i.e.\ a
second moment that diverges for large $k_{\mathrm{max}}$, this would therefore
be expected not to exhibit threshold behaviour.

The DHSC can be used to sample efficiently from the final size distribution of
the epidemic. Figure~\ref{fig:fs} shows the results of this simulation, for
constant and exponential $\Gamma$ as well as $\alpha = -2$, $-3$, $-4$,
$-\infty$, when the natural choice $k_{\mathrm{max}} = N-1$ is made. This
visualises how low values of $|\alpha|$ are associated with less
critical-looking behaviour at $\rho=1$. Note that this is not in contradiction
with Theorem~\ref{odelim}, and the analysis of the deterministic limit
threshold, which requires finite $\mathcal{K}$. In the current context this
means that $k_{\mathrm{max}} \ll N$, which does not hold in the simulations and
highlights the problems of reliance on ODE methods alone.

In conclusion, it is possible to provide a mathematically transparent model
of extreme discrete heterogeneity, which allows fast sampling from the
final size distribution to allow insights into the subtle epidemiological
phenomena associated with extreme heterogeneity in mixing.

\vspace*{0.5cm}
\section*{Acknowledgements}

Work funded by the UK Engineering and Physical Sciences Research Council. 

\clearpage

\begin{figure}[H]
\begin{center}
{\resizebox{\textwidth}{!}{ \includegraphics{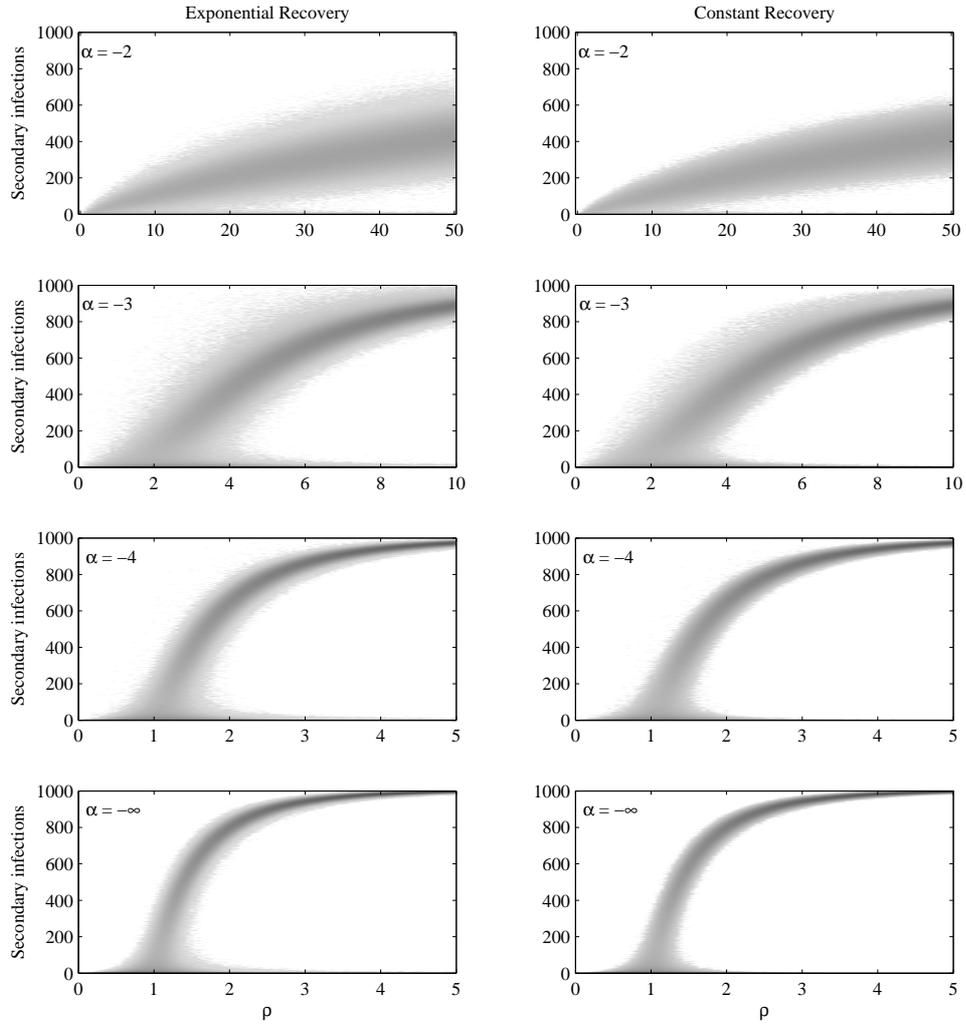} }}
\end{center}
\caption{Numerical results for the final size distribution for exponential mean
1 and constant unit infectious period, truncated Zipf degree distribution, a
total population size $N=10^3$, $m=1$ initial infective, and $10^5$ realisations
for each value of $\alpha$, $\rho$, and distribution. Colour intensity
is proportional to probability${}^{1/6}$}
\label{fig:fs}
\end{figure}

\clearpage

\appendix

\section{Code for the DH Sellke construction}

The following MATLAB function accepts a vector of degree numbers \texttt{K}, a
number of initial infectives \texttt{m}, a transmission rate \texttt{tau},
and a recovery rate \texttt{gamma}. If the final parameter is set to
\texttt{NaN}, a fixed unit recovery time is used. A sample from the final size
distribution is returned.\\

\begin{verbatim}
function Z = het_sel(K,m,tau,gamma)
  N = length(K);
  if ~isnan(gamma)
    T = exprnd(1/gamma,1,N);
  else
    T = ones(1,N);
  end
  [Q, ii] = sort(exprnd(1./K((m+1):N)));
  La0 = sum(tau*T(1:m).*K(1:m));
  La = [La0, La0+cumsum(tau*T(m+ii).*K(m+ii))];
  z = find(Q > La(1:(end-1)), 1);
  if isempty(z)
    Z = N;
  else
    Z = z+m-1;
  end
end
\end{verbatim}

\end{document}